\documentclass[journal]{IEEEtran}

\usepackage{myStyleIEEE}
\usepackage{nomencl}
\IEEEoverridecommandlockouts

\usepackage{etoolbox}
\renewcommand\nomgroup[1]{%
  \item[\bfseries
  \ifstrequal{#1}{P}{A. Parameters}{%
  \ifstrequal{#1}{V}{C. Variables}{%
  \ifstrequal{#1}{S}{B. Sets and Indices}{}}}%
]}

\begin{document}
\bstctlcite{IEEEexample:BSTcontrol}

\title{Towards Optimal System Scheduling with Synthetic Inertia Provision from Wind Turbines}

\newtheorem{proposition}{Proposition}
\renewcommand{\theenumi}{\alph{enumi}}

\newcommand{\uros}[1]{\textcolor{magenta}{$\xrightarrow[]{\text{Uros}}$ #1}}

\author{Zhongda~Chu,~\IEEEmembership{Student~Member,~IEEE,}
        Uros~Markovic,~\IEEEmembership{Member,~IEEE,}
        Gabriela~Hug,~\IEEEmembership{Senior~Member,~IEEE,} 
        Fei~Teng,~\IEEEmembership{Member,~IEEE} 
        \thanks{Z. Chu and F. Teng (e-mail: f.teng@imperial.ac.uk) are with Imperial College London; 
        U. Markovic and G. Hug are with ETH Zurich.
        
        This work was supported by EPSRC under Grant EP/T021780/1 and by ESRC under Grant ES/T000112/1.}
    }
\maketitle
\IEEEpeerreviewmaketitle

\begin{abstract}
The undergoing transition from conventional to converter-interfaced renewable generation leads to significant challenges in maintaining frequency stability due to declining system inertia. In this paper, a novel control framework for Synthetic Inertia (SI) provision from Wind Turbines (WTs) is proposed, which eliminates the secondary frequency dip and allows the dynamics of SI from WTs to be analytically integrated into the system frequency dynamics. Furthermore, analytical system frequency constraints with SI provision from WTs are developed and incorporated into a stochastic system scheduling model, which enables the provision of SI from WTs to be dynamically optimized on a system level. Several case studies are carried out on a Great Britain 2030 power system with different penetration levels of wind generation and inclusion of frequency response requirements in order to assess the performance of the proposed model and analyze the influence of the improved SI control scheme on the potential secondary frequency dip. The results demonstrate that the inclusion of SI provision from WTs into Unit Commitment (UC) can drastically impact the overall system costs.
\end{abstract}

\begin{IEEEkeywords}
synthetic inertia, frequency response, unit commitment, recovery effect, secondary frequency dip
\end{IEEEkeywords}

\makenomenclature
\mbox{}
\nomenclature[P]{$J$}{Lumped inertia of WT driven systems$\,[\mathrm{Mkgm}^2]$}
\nomenclature[P]{$\rho$}{Air density$\,[\mathrm{kg/m}^3]$}
\nomenclature[P]{$R$}{WT radius$\,[\mathrm{m}]$}
\nomenclature[P]{$v_w$}{Wind speed$\,[\mathrm{m/s}]$}
\nomenclature[P]{$T_d$}{Fully delivered time of PFR$\,\mathrm{[s]}$}
\nomenclature[P]{$\Delta \dot f_\mathrm{lim}$}{Maximum permissible RoCoF$\,[\mathrm{Hz/s}]$}
\nomenclature[P]{$\Delta f_\mathrm{lim}$}{Maximum permissible frequency deviation$\,[\mathrm{Hz}]$}
\nomenclature[P]{$\omega_{r,\mathrm{min}}$}{WT minimum allowable rotor speed$\,[\mathrm{rad/s}]$}
\nomenclature[P]{$H_g$}{Inertia time constant of unit $g\,\mathrm{[s]}$}
\nomenclature[P]{$P_g^\mathrm{max}$}{Installed capacity of unit $g\,\mathrm{[MW]}$}
\nomenclature[P]{$f_0$}{Nominal frequency$\,\mathrm{[Hz]}$}
\nomenclature[P]{$t_n$}{Time instance of frequency nadir$\,\mathrm{[s]}$}
\nomenclature[P]{$\pi(n)$}{Probability of scenario $n$}
\nomenclature[P]{$C_g(n)$}{Operating cost associated with unit $g$ of scenario n}
\nomenclature[P]{$P_{wj}^{c}$}{Capacity of the $j$-th WF$\,\mathrm{[MW]}$}
\nomenclature[P]{$\lambda$}{WT tip ratio}
\nomenclature[P]{$\kappa_j(v_w)$}{Wind speed distribution in the $j$-th WF}
\nomenclature[P]{$\omega_{r,0}$}{Initial WT rotor speed before the disturbance$\,[\mathrm{rad/s}]$}
\nomenclature[P]{$\theta$}{WT pitch angle$\,[\mathrm{rad}]$}
\nomenclature[P]{$\Tilde{P}_{a,0}$}{WT pre-disturbance mechanical power$\,[\mathrm{MW}]$}
\nomenclature[P]{$D$}{Load-dependent damping in the system$\,[\mathrm{MW/Hz}]$}
\nomenclature[P]{$\Delta P_L$}{Loss of generation$\,\mathrm{[MW]}$}

\nomenclature[V]{$\omega_r$}{WT rotor speed$\,[\mathrm{rad/s}]$}
\nomenclature[V]{$P_a$}{Aerodynamic power of WTs$\,[\mathrm{MW}]$}
\nomenclature[V]{$P_e$}{Electrical power of WTs$\,[\mathrm{MW}]$}
\nomenclature[V]{$C_p$}{Aerodynamic power coefficient}
\nomenclature[V]{$\Delta \tilde{P}_a$}{Output power of MPE$\,[\mathrm{MW}]$}
\nomenclature[V]{$\Delta {P}_{opt}$}{WT MPPT setpoint variation$\,[\mathrm{MW}]$}
\nomenclature[V]{$P_{SI}$}{WT SI control output power$\,[\mathrm{MW}]$}
\nomenclature[V]{$\Delta {P}_e^*$}{WT electrical power reference variation$\,[\mathrm{MW}]$}
\nomenclature[V]{$\Delta {P}_e$}{WT electrical power variation$\,[\mathrm{MW}]$}
\nomenclature[V]{$H_s$}{WT synthetic inertia$\,[\mathrm{MWs/Hz}]$}
\nomenclature[V]{$\Delta f$}{System frequency deviation$\,[\mathrm{Hz}]$}
\nomenclature[V]{$\Delta \hat P_a$}{Linear approximation of $\Delta \tilde{P}_a$ $\,[\mathrm{MW}]$}
\nomenclature[V]{$D_s$}{Damping coefficient in MPE approximation$\,[\mathrm{MW/Hz}]$}
\nomenclature[V]{$H_{s,\mathrm{max}}$}{Maximum synthetic inertia of WTs$\,[\mathrm{MWs/Hz}]$}
\nomenclature[V]{$P_0$}{WT initial active power output$\,[\mathrm{MW}]$}
\nomenclature[V]{$P_{\mathrm{max}}$}{WT maximum active power output$\,[\mathrm{MW}]$}
\nomenclature[V]{$H_{s_j}^C$}{SI capacity of the $j$-th WF$\,[\mathrm{MWs/Hz}]$}
\nomenclature[V]{$H_{s,\mathrm{max}_j}$}{Maximum available SI of a WT in the $j$-th WF $[\mathrm{MWs/Hz}]$}
\nomenclature[V]{$\Delta \hat P_{a_j}$}{Mechanical power loss of the $j$-th WF$\,[\mathrm{MW}]$}
\nomenclature[V]{$\Delta \hat P_{a_{j}}^0$}{Mechanical power loss in a single WT of the $j$-th WF$\,[\mathrm{MW}]$}
\nomenclature[V]{$H_{s_j}$}{Synthetic inertia of the $j$-th WF$\,[\mathrm{MWs/Hz}]$}
\nomenclature[V]{$H_c$}{Synthetic inertia of conventional synchronous generators $[\mathrm{MWs/Hz}]$}
\nomenclature[V]{$\Delta R(t)$}{Power injection through the Primary Frequency Response (PFR)$\,\mathrm{[MW]}$}
\nomenclature[V]{$N_g^\mathrm{up}$}{Operation status of unit $g$ with $0/1$ indicating unit being offline/online }
\nomenclature[V]{$H$}{Total inertia in the system$\,[\mathrm{MWs/Hz}]$}
\nomenclature[V]{$\Bar{H}_s$}{Total synthetic inertia in the system$\,[\mathrm{MWs/Hz}]$}
\nomenclature[V]{$\Delta \bar{P}_a$}{Total mechanical power loss due to the SI provision $[\mathrm{MW}]$}
\nomenclature[V]{$\Delta P_w$}{Total WT power injection to the system$\,[\mathrm{MW}]$}
\nomenclature[V]{$\pmb{\mathrm{H_s}}$}{SI time constant$\,[\mathrm{s}]$}

\nomenclature[S]{$j\in\mathcal{F}$}{Set of WFs}
\nomenclature[S]{$g\in\mathcal{G}$}{Set of conventional generators}
\nomenclature[S]{$i\in\mathcal{P}$}{Set of planes in piece-wise linearization}
\nomenclature[S]{$n\in \mathcal{N}$}{Set of nodes in the scenario tree}
\printnomenclature

\section{Introduction} \label{sec:1}
Zero carbon operation of the Great Britain (GB) power system is expected to be achieved by 2025 in order to reduce overall greenhouse gas emissions \cite{zero-c}. Most notably, the system operator is accounting for large-scale integration of wind generation \cite{10-year}. With the conventional power plants being replaced by wind turbines, significant challenges are anticipated in terms of system operation and stability \cite{8450880}. One such challenge is driven by the reduction of system inertia, since wind turbines are interfaced to the grid through power electronic converters that decouple the rotational Kinetic Energy (KE) from the system. It is predicted that the total system inertia in the UK will be reduced by up to 70\% by 2033/34\cite{FES,8830447}.

In order to maintain secure and stable system operation, various control strategies have been developed to facilitate the provision of frequency support from Variable Speed Wind Turbines (VSWTs). In general, two control techniques have been proposed to increase the power injection from VSWTs into the grid during frequency events: deloading and overproduction \cite{DREIDY2017144}, with the latter being more popular as it maintains the optimal energy utilization under normal operating conditions \cite{8064700}. During overproduction additional power is injected into the grid from VSWTs through either pitch control at above-rated wind speeds or KE extraction at below-rated wind speeds. For pitch control, the maximum power injection is restricted by the pitch angle adjustment rate and the converter capacity rating. In general, an increase by 0.2~$\mathrm{p.u.}$ above the rated power for 10 seconds is achievable \cite{KARBOUJ2019488}. On the other hand, KE extraction involves rotor deceleration, thus limited by the minimum permissible rotor speed to ensure mechanical stability. Furthermore, the rotor deceleration drives the operating point away from its optimal value set by the Maximum Power Point Tracking (MPPT) control, leading to the decrease of the mechanical power captured from wind. Therefore, an underproduction period is inevitable after the SI provision to cease further rotor deceleration, which is the major limitation of SI provision through KE extraction \cite{ATTYA20182071}.

A number of control schemes have been proposed to mitigate the underproduction and the associated secondary frequency dip. Instead of a step decrease in frequency support, a ramp power reduction is employed at the end of the overproduction period in \cite{HAFIZ2015629,6038914}. Alternatively, an improved strategy is presented in \cite{8064700} by navigating the electrical power output of a WT in an optimal fashion from the mechanical curve to the MPPT trajectory, whereas \cite{6423240} suggests a rotor speed-dependent inertia provision to alleviate the power reduction at the end of a frequency support period. An energy storage system is also considered to compensate the power decrease due to frequency support termination \cite{7017580}. All of the aforementioned approaches reduce to some extent the step-change in WT's power, but they also neglect the mechanical power reduction due to the loss of efficiency which can be significant depending on the rotor speed deviation. \textcolor{black}{Moreover, most of the SI control applications for WTs in the literature focus on the local device-level optimization where the Phase-Locked Loop (PLL) is usually applied to measure the system frequency and its derivative for the purposes of SI control. However, the PLL introduces additional control loops and measurement delay. The authors in \cite{7797125} point out that the PLL dynamics would increase the oscillations and settling time of the frequency deviation. On the other hand, the derivative control used for obtaining RoCoF also makes it sensitive to noise and can lead to unstable operation~\cite{app7070654,8611073}. To address the potential instability of PLLs and avoid the high frequency noises introduced by differential operators, \cite{8468053} proposes a noise-free estimation of the frequency derivative using frequency-locked loops. Regardless of the potential drawbacks related to the derivative term, the optimal SI control can still be achieved on the device-level. However, the optimal performance for a single WT may not correspond to the best solution for the entire system.}

On the other hand, recent work on UC has revealed the importance of including post-disturbance frequency dynamics within the scheduling process. More precisely, the analytical expressions for frequency nadir and Rate-of-Change-of-Frequency (RoCoF) can be explicitly included as constraints in the UC optimization problem. However, the frequency nadir term is highly nonlinear and not suitable for traditional UC formulations, often solved as a Mixed-Integer Linear Program (MILP). This has been addressed in \cite{6717054,7115982,8171772,Matthieu2019} through various forms of linearization of the respective constraint. Additionally, in \cite{Matthieu2019} the state-of-the-art converter control schemes of inverter-based generation are included in the system frequency dynamics and virtual inertia units are explicitly considered for inertia provision.  \cite{8667397} and \cite{Badesa_2019} consider the combination of a finite number of frequency response services with different delivery times and a dynamically-reduced largest power in-feed. However, due to its distinguishing characteristics, the SI provision from WTs leads to a different and more complex system dynamic evolution during a frequency event. As such, these dynamics should be explicitly modeled in the scheduling process, i.e., dispatch and unit commitment decisions, in order to optimize SI provision and achieve minimum system operation cost while maintaining the system frequency constraints. 
The authors in \cite{7366764} investigate optimal frequency support by WTs from the system perspective, which however relies on dynamic simulations for different system conditions and therefore cannot be directly incorporated into the system scheduling model. The study in \cite{7370811Teng} investigates the impact of WTs and their SI provision on system costs and highlights the need for dynamic optimization of SI services with consideration of the underproduction period. Nevertheless, the detailed WT dynamics have not been considered in this work.


In the vein of previous research, this paper proposes a novel system scheduling model to optimize the SI provision from WTs with the objective to minimize system operation cost, while simultaneously keeping the system frequency response within prescribed limits, as well as preserving the mechanical stability of WTs through explicit modeling of the detailed turbine dynamics. The key contributions are identified as:
\begin{enumerate}
    \item 
    A novel control framework is applied for SI provision from WTs, which accounts for the WT underproduction and eliminates the associated secondary frequency dip. Under such control framework, the WT underproduction can be approximated as a negative system damping term and the conservativeness of such approximation is proved.
    
    \item The system frequency constraints are derived including detailed WT dynamics under the SI control framework. The highly nonlinear, $n$-dimensional hyperboloid constraint for frequency nadir is efficiently linearized, with guaranteed conservativeness and quantified accuracy.

    \item The proposed system frequency constraints are integrated into an MILP based stochastic scheduling model to assess the impact of optimized SI provision and WT underproduction. Furthermore, the maximum SI provision from WTs in the system is estimated and constrained to ensure feasibility and mechanical stability of the wind farm.
\end{enumerate}
 
The remainder of this paper is structured as follows. Section~\ref{sec:2} describes the VSWT model and the proposed SI control framework. The system frequency dynamics and the derivation of the respective frequency constraints under the proposed control framework are presented in Section~\ref{sec:3}, together with a linearization that allows the integration of the constraint into a Stochastic UC (SUC) formulation. Section~\ref{sec:4} showcases the results and validates the performance of the proposed model on several case studies, whereas Section~\ref{sec:5} concludes the paper.
    
\section{Synthetic Inertia Control Design} \label{sec:2} 
\subsection{VSWT Modeling} \label{sec:2.1} 
A simplified wind turbine model can be represented by a mechanical model, an aerodynamic model and a Maximum Power Point Tracking (MPPT) model, while the dynamics of the electrical system can be neglected due to its shorter timescales compared to the mechanical system \cite{7393806}. The mechanical model can be approximated by the dynamics of a single rotating mass as:
\begin{equation}
    \label{TD0}
    J \dot \omega_r \omega_r = P_a - P_e .
\end{equation}
The aerodynamic power can be calculated using the following model:
\begin{equation}
    \label{P_aero}
    P_a = \underbrace{\frac{\pi}{2\cdot10^6}\rho R^2 v_w^3}_{\eta_a} C_p (\lambda, \theta) ,
\end{equation}
with $C_p(\lambda, \theta)$ denoting the power coefficient dependent on the pitch angle $\theta$ and the tip ratio $\lambda = \omega_r R/v_w$. Since the deloading control strategy is not considered in this paper, it is assumed that WTs operate in MPPT mode, thus producing maximum power corresponding to an optimal $C_p$ given the current wind speed and pitch angle. The formulation proposed in \cite{970114} is applied for calculating the power coefficient:
\begin{align}
\begin{split}
    C_p(\lambda,\theta) &= 0.22\left(\frac{116}{\lambda_i}-0.4\theta-5\right)e^{-\frac{12.5}{\lambda_i}},\\
    \frac{1}{\lambda_i} &= \frac{1}{\lambda+0.08\theta}-\frac{0.035}{\theta^3+1}. \label{cp2}
\end{split}
\end{align}


\begin{figure}[!b]
    \centering
	\scalebox{0.735}{\includegraphics[]{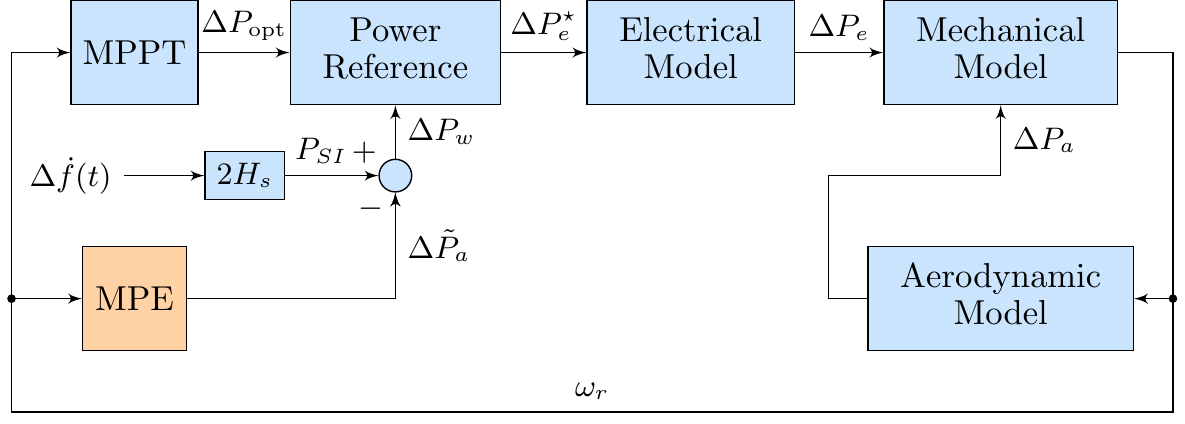}}
    \caption{\label{fig:Block}Block diagram of KE extraction control scheme.}
\end{figure}

\subsection{SI Control Scheme} \label{sec:2.2} 


In this subsection, a novel control framework is designed for SI provision from WTs while eliminating the secondary frequency dip. After the detection of a large disturbance, the MPPT mode is deactivated and WTs employ stored KE to inject additional active power $P_\mathrm{SI}$ proportional to the RoCoF signal $\Delta \dot{f}$ into the grid. Consequently, the rotor starts to decelerate and deviate from its optimal operating point, causing the loss of mechanical power captured from the wind. This power loss is conventionally hidden from the system during the frequency support period and tends to appear as a step disturbance afterwards leading to a secondary frequency dip, which can be problematic for a low-inertia system. Furthermore, from the system scheduling perspective, it is critical to quantify such secondary disturbances and schedule adequate resources to maintain the frequency constraints. This is very challenging within existing control frameworks due to the coupling of the mechanical and the aerodynamic models.
\begin{figure}[!t]
    \centering
	\scalebox{1.1}{\includegraphics[]{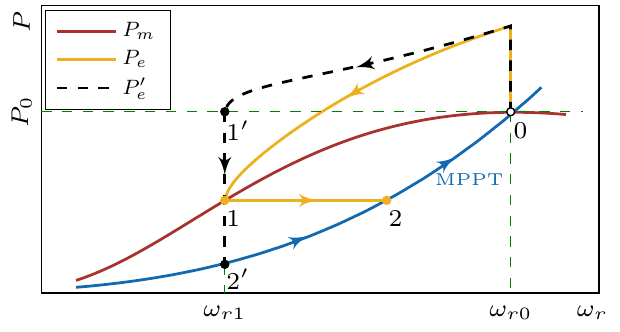}}
    \caption{\label{fig:Trajectory}Operating point trajectory of WT with SI control during a frequency disturbance event.}
\end{figure}

\begin{figure}[!t]
    \centering
	\scalebox{1.1}{\includegraphics[]{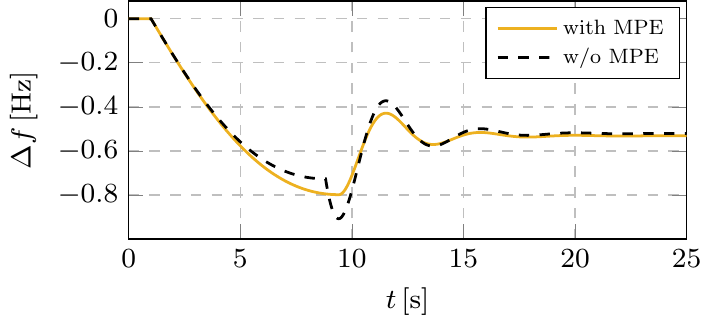}}
    \caption{\label{fig:f_MPE}System frequency evaluation under different SI control schemes.}
\end{figure}

In order to achieve the desired performance, a Mechanical Power Estimator (MPE) as first proposed in \cite{5275387} is included in the control framework, which captures the mechanical power variation according to the measured rotor speed. Its output $\Delta \tilde{P}_a$ adjusts the SI control feedback $P_\mathrm{SI}$, and combined with the MPPT setpoint $\Delta P_\mathrm{opt}$ yields the new electrical power reference $\Delta P_e^\star$. The equivalent block diagram is illustrated in Fig.~\ref{fig:Block}, where the $\Delta$ symbols denote the deviations from the values before the disturbance. In this way, the gradual decrease in mechanical power is constantly observed by the system and there is no sudden power change resulting from the termination of SI provision at the end of the KE extraction process. 

The trajectory of the electrical power reference with respect to the rotor speed is demonstrated in Fig.~\ref{fig:Trajectory}. Starting at the pre-disturbance steady-state point $(0)$, the usual overproduction operation (denoted in black) is now improved with the addition of MPE (denoted in yellow), thus resulting in the elimination of the step power change $(1'-2')$ and the associated secondary frequency dip. After the detection of the frequency nadir, the recovery process $(1-2-0)$ is enabled, characterized by a constant electrical power until the MPPT operation is restored. It should be noted that since point $(1)$ is an equilibrium (i.e., mechanical power is equal to electrical power), a small disturbance would trigger the turbine speed recovery. In such case, the acceleration power would result in longer recovery time and higher wind energy loss, which is however negligible on a timescale of primary frequency response. While further optimization to achieve a balance between additional power supply and speed of recovery during the secondary frequency response is possible as demonstrated in \cite{8064700}, it is less relevant for the problem at hand and therefore not addressed in this paper. 


\textcolor{black}{To have a better understanding of the proposed SI control scheme and the influence of MPE, an example of system frequency evaluation is presented in Fig.~\ref{fig:f_MPE}. With the inclusion of MPE in the SI control, the system frequency deviation becomes slightly larger as less power is injected to the grid. However, a significant secondary frequency dip is observed in the case without MPE due to the sudden power loss after the SI provision (i.e., trajectory 1'-2' in Fig.~\ref{fig:Trajectory}), leading to a worse frequency nadir overall.}

Using the above mentioned procedure, the rotor speed deviation can be expressed analytically as follows:
\begin{equation}
    \label{TD1}
        J \dot \omega_r  \omega_r  = \Delta P_a  - \Bigg(\underbrace{\Delta \tilde{P}_a  - 2H_{s}\frac{\partial\Delta f }{\partial t}}_{\Delta P_e}\Bigg)= 2H_{s}\frac{\partial\Delta f }{\partial t}.
\end{equation}
\textcolor{black}{It is assumed that the MPE error is negligible due to several reasons, namely: (i) a reasonably accurate rotor speed measurement obtained through either a sensor-based or a sensorless technique is always available for WT control; and (ii) the aerodynamic model described by (2) is widely used in the literature with satisfactory performance~\cite{1198317,20030251,4282019}. } 
The relationship between $\Delta f$ and $\omega_r$ can be derived by taking the integral of \eqref{TD1} on both sides:
\begin{equation}
\label{w_r}
    \omega_r(\Delta f, H_s) = \sqrt{\frac{4 H_s\Delta f}{J}+\omega^2_{r,0}},
\end{equation}
where $\omega_{r,0}$ is the initial rotor speed before the disturbance. \textcolor{black}{By substituting \eqref{w_r} into \eqref{P_aero}-\eqref{cp2}, the loss of generation due to the deviation from the maximum power point can be expressed as a function of system frequency deviation instead of WT rotor speed:}

\begin{equation}
\label{P_aero(w_r)}
\begin{split}
    \Delta \Tilde{P}_a = 0.22\eta_a\left(\frac{116}{\lambda}-9.06\right)e^{0.4375-\frac{12.5}{\lambda}}-  \Tilde{P}_{a,0},
\end{split}
\end{equation}
where $\Delta f$ is subsumed in $\lambda$  and the pitch angle $\theta$ is set to zero, due to the below-rated wind speed in KE extraction control scheme. \textcolor{black}{It should be noted that pitch control is used at higher wind speeds. Under such control, a specified level of virtual inertia can be provided to the grid without any other impact on the system frequency response. Therefore, the details of the pitch control are not included in the paper.}

The proposed control framework ensures that the power injection from WTs during the frequency support period is a combination of the RoCoF control and the MPE feedback. This process continues until the frequency nadir is reached (corresponding to instantaneous RoCoF being zero).
At this time instance, the WT mechanical power equals the electrical power output, resulting in a constant rotor speed followed by a recovery period as discussed in Section \ref{sec:2.2}.

\subsection{MPE Approximation} \label{sec:2.3} 
The convoluted expression for $\Delta\Tilde{P}_a(t)$ described in \eqref{P_aero(w_r)} makes it hard to obtain an explicit expression for system frequency evolution upon inclusion in the frequency dynamics scheme. Therefore, a linear approximation is introduced instead, which can be interpreted as an additional damping term of the form:
\begin{equation}
    \label{Ds1}
    \Delta \hat P_a(\Delta f, H_s) = D_s(H_s)\Delta f,
\end{equation}
with the damping coefficient $D_s$ defined as
\begin{equation}
    \label{Ds2}
    D_s(H_s) = \frac{\Delta \Tilde P_a(-\Delta f_\mathrm{lim}, H_s)}{-\Delta f_\mathrm{lim}}.
\end{equation}
Due to its specific structure, the relationship in \eqref{Ds2} can be fitted with a quadratic function:
\begin{equation}
    \label{Ds3}
    D_s(H_s) = \underbrace{\frac{D_s(H_{s,\mathrm{max}})}{H^2_{s,\mathrm{max}}}}_{\gamma} H_s^2,
\end{equation}
where $H_{s,\mathrm{max}}$ denotes the maximum synthetic inertia WTs can provide for a given wind speed forecast. An example of the proposed approximation is illustrated in Fig.~\ref{fig:plots}. While the quadratic expression in \eqref{Ds3} captures the damping accurately, the approximation of $\Delta\Tilde{P}_a(\Delta f, H_s)$ has a noticeable mismatch as a consequence of the error in \eqref{Ds1}. Nonetheless, it will always be on the conservative side as proved below.

\begin{figure}[!t]
  \centering
    \begin{minipage}{0.5\textwidth}
        \centering
        \scalebox{1.1}{\includegraphics[]{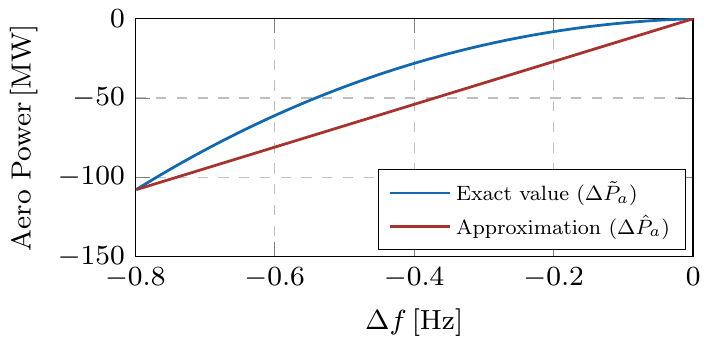}} 
        \vspace{0.15cm}
    \end{minipage} 
    \begin{minipage}{0.5\textwidth}
        \centering
        \scalebox{1.1}{\includegraphics[]{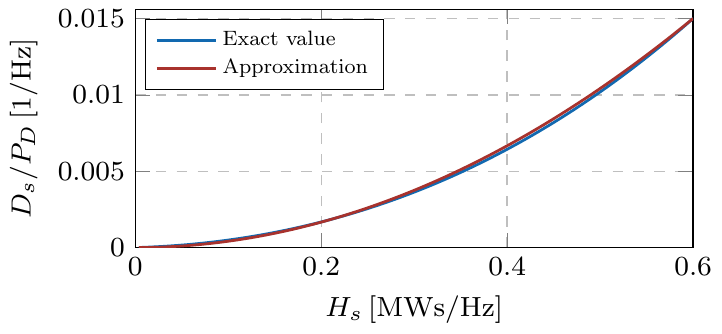}} 
    \end{minipage} 
  \caption{\label{fig:plots}Approximation of MPE output: (i) aerodynamic power $\Delta \hat{P}_a$; (ii) damping term in p.u. $D_s/P_D$.}
\end{figure}

\begin{proposition}
$\Delta \Tilde P_a(\Delta f)\ge \Delta \hat P_a(\Delta f),\forall \Delta f \in [-\Delta f_\mathrm{lim}, 0]$. The equality holds if and only if $\Delta f = -\Delta f_\mathrm{lim} $ or $\Delta f = 0 $.
\end{proposition}

\begin{proof}
According to \eqref{w_r}, $\Delta f= 0$ implies $\omega_r = \omega_{r,0}$ leading to:
\begin{equation}
    \label{Pa(0)}
    \Delta \Tilde P_a(0)= \Delta \hat P_a(0)=0.
\end{equation}
Equation \eqref{Ds2} results in the following:
\begin{equation}
    \label{Pa(f_lim)}
    \Delta \Tilde P_a(-\Delta f_\mathrm{lim})= \Delta \hat P_a(-\Delta f_\mathrm{lim})=-D_s \Delta f_\mathrm{lim}.
\end{equation}
The second-order derivative of $\Delta \Tilde P_a$ with respect to $\Delta f$ can be derived by applying the chain rule as follows:
\begin{equation}
\begin{split}
    \frac{\partial^2 \Delta \Tilde P_a}{\partial \Delta f^2} =& \frac{\partial^2 \Delta \Tilde P_a}{\partial \omega_r ^2}\left(\frac{\partial \omega_r}{\partial \Delta f}\right)^2 + \frac{\partial \Delta \Tilde P_a}{\partial \omega_r}\frac{\partial^2 \omega_r}{\partial \Delta f^2}\\
    =& g_1(\omega_r) g_2(\omega_r),
\end{split}
\end{equation}
where 
\begin{subequations}
\begin{align}
    g_1(\omega_r) =& \frac{4 \eta_a H_s^2 v_w^2}{J^2 R^2 \omega_r^6} e^{0.4375-\frac{12.5v_w}{R\omega_r}},\\
    g_2(\omega_r) =& \frac{18125v_w}{R\omega_r}+\frac{188 R\omega_r}{v_w}-8665.625.
\end{align}
\end{subequations}
It is straightforward to show that $g_1(\omega_r)>0$. Given the WT operating conditions, the tip ratio $\lambda = \frac{R \omega_r}{v_w}$ is kept within the $(0.21, 45.88)$ range, thus resulting in $g_2(\omega_r)<0$. Therefore, $\frac{\partial^2 \Delta \Tilde P_a}{\partial \Delta f^2}<0$, which together with \eqref{Pa(0)} and \eqref{Pa(f_lim)} concludes the proof.
\end{proof}

\subsection{Aggregated SI from Wind Farm}  \label{sec:2.4} 
During the system scheduling process, it is difficult for the system operator to consider each WT individually. Therefore, in this section an aggregated Wind Farm (WF) model is proposed where the total available synthetic inertia (termed SI capacity hereafter) and the corresponding mechanical power loss $(\Delta \hat{P}_a)$ of a WF are derived based on the wind speed distribution. Moreover, it is assumed that the WTs in a WF can be represented by an identical dynamical model and that the wind speed distribution is available to the system operator. Note that if one considers individual WTs and wind speed information of the entire system, the following discussion refers to an aggregated SI of the system.

\subsubsection{SI Capacity Estimation}\label{sec:2.4.1}
In order to ensure mechanical stability of each WT, it is necessary to maintain rotor speed above its minimum permissible level. Similarly, the SI capacity provided by a WT should be within its physical capabilities. According to the energy conservation law, the expression for SI extracted from a turbine can be derived from \eqref{w_r} as a function of the initial operating point:
\begin{equation}
\label{H_s}
    H_{s} = \frac{J\left(\omega^2_{r}-\omega^2_{r,0}\right)}{4\Delta f}.
\end{equation}
In addition, the installed converter capacity also limits the maximum inertia provision at near-rated operating points. Therefore, the maximum available SI capacity $H_{s,\mathrm{max}}$ of a single WT can be determined as follows:
\begin{equation}
\label{H_s,max}
    H_{s,\mathrm{max}} = \min\left\{ \frac{J\left(\omega^2_{r,0} - \omega^2_{r,\mathrm{min}}\right)}{4\left|\Delta f_\mathrm{lim}\right|}, \frac{P_\mathrm{max}-P_0}{2\big|\Delta \dot f_\mathrm{lim}\big|}\right\}.
\end{equation}

The initial rotor speed is regulated according to the current wind speed through MPPT control, i.e., $\omega_{r,0} = \omega_r(v_w)$. Based on the condition in \eqref{H_s,max}, the relationship between $H_{s,\mathrm{max}}$ and $v_w$ can be derived, as illustrated in Fig.~\ref{fig:Hs,m}. One may notice that the KE extraction mode is initially characterized by the increase of $H_{s,\mathrm{max}}$ with wind speed as more KE becomes available. However, the SI capacity starts to decay with the wind speed approaching the rated value and remains constant thereafter due to the limit on maximum WT power capacity. Having obtained $H_{s,\mathrm{max}}(v_w)$, the total SI capacity $H_{s_j}^C$ of a WF can be computed based on its respective wind speed distribution inside the wind farm:
\begin{equation}
\label{H^C}
    H_{s_j}^C = \int_{0}^{\infty} N_j H_{s,\mathrm{max}_j}(v_{w}) \kappa_j(v_w) dv_{w}.
\end{equation}

\begin{figure}[!t] 
	\centering
	\scalebox{1.1}{\includegraphics[]{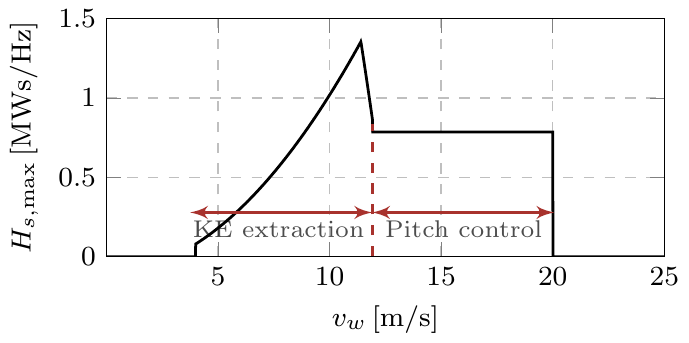}}
	\caption{Evolution of maximum SI capability of a single WT with wind speed for KE extraction and pitch control modes.}
	\label{fig:Hs,m}
\end{figure}
\subsubsection{Mechanical Power Loss}
Similarly, the total mechanical power loss $\Delta \hat P_{a_j}$ of the $j$-th WF can be calculated as:
\begin{equation}
\label{Pa_j}
    \Delta \hat P_{a_j} = \int_{0}^{\infty} N_j \Delta\hat{P}_{a_{j}}^0(v_{w}) \kappa_j(v_w) dv_{w}.
\end{equation}
Following the same procedure as in Section \ref{sec:2.3}, the aggregated mechanical power loss can be rewritten in the form of system damping:
\begin{equation}
    \label{WF_D}
    \Delta \hat P_{a_j} = D_{s_j} \Delta f = (\gamma_j H_{s_j}^2) \Delta f.
\end{equation}
Although the optimal allocation of $H_{s_j}$ to each individual WT is out of the scope of this paper, it can be shown that there always exist a realization such that the actual mechanical power loss is less than or equal to \eqref{WF_D}.

\section{System Frequency Stability Constraints with SI Provision under SUC} \label{sec:3}
\subsection{System Frequency Modeling with WT SI Control} \label{sec:3.1}
Under the premise of the Centre-of-Inertia (CoI) model, frequency dynamics in a multi-machine power system can be expressed in the form of a single swing equation \cite{7833096}:
\begin{equation}
    \label{sw1}
    2H_c\frac{\partial\Delta f(t)}{\partial t} = -D \Delta f(t) + \Delta R(t) -\Delta P_L,
\end{equation}
where $\Delta P_L$, the loss of a generator at $t=0$ can be viewed as a step disturbance. Moreover, the PFR $\Delta R(t)$ can be represented according to the following scheme \cite{6714513}:
\begin{equation}
\label{R}
\Delta R(t)=
     \begin{cases}
       \frac{R}{T_d}t &, \; 0\le t< T_d \\ 
       R &, \; T_d\le t
     \end{cases}.
\end{equation}
The total inertia of conventional generators is computed as:
\begin{equation}
    \label{H1}
    H_c =\frac{\sum_{g\in \mathcal{G}} H_g  P_g^\mathrm{max} N_g^\mathrm{up}}{f_0}.
\end{equation}

Incorporating the proposed SI control framework into \eqref{sw1} yields:
\begin{equation}
    \label{sw2}
    2H\frac{\partial\Delta f(t)}{\partial t} = -D \Delta f(t) + \Delta R(t) -\Delta P_L + {\Delta \bar{P}_a}, 
\end{equation}
with $H = H_c + \Bar{H}_s= H_c+ \sum_{j\in \mathcal{F}} H_{s_j}$ and $\Delta \bar{P}_a = \sum_{j\in \mathcal{F}} \Delta \hat P_{a_j}$ reflecting the total inertia and the mechanical power loss due to the SI provision in the system respectively. Note that $\Delta \bar{P}_a$ is negative for an under frequency event because of the formulation in \eqref{P_aero(w_r)}. 
Combining \eqref{WF_D} and \eqref{sw2} gives:
\begin{equation}
    \label{sw3}
    2H\frac{\partial\Delta f(t)}{\partial t} = -\Big(\underbrace{D- \sum_{j\in \mathcal{F}} \gamma_j H_{s_j}^2}_{D'}\Big) \Delta f(t) + \Delta R(t) -\Delta P_L,  
\end{equation}
where the effect of WT rotor deceleration is now modeled as a negative damping, as a function of SI provided to the system. From \eqref{sw3}, the mathematical expressions for maximum instantaneous RoCoF $(\Delta \dot f_\mathrm{max}\equiv\Delta \dot f|_{t=0^+})$ and steady-state frequency deviation $(\Delta f_\mathrm{max}^\mathrm{ss}\equiv\Delta f|_{t=\infty})$ are obtained:
\begin{equation}
\label{rocof,fss}
    \Delta \dot f|_{t=0^+} = -\frac{\Delta P_L}{2H} , \; \Delta f|_{t=\infty} = \frac{R-\Delta P_L}{D'}.
\end{equation}
Based on the magnitude of disturbance $\Delta P_L$, both metrics should be kept within prescribed limits by selecting appropriate $H$ and $R$ terms, respectively. Furthermore, the time-domain solution for frequency deviation can be derived by substituting \eqref{R} into \eqref{sw3} as follows:
\begin{equation}
\label{f(t)}
    \Delta f(t) = \left(\frac{\Delta P_L}{D'}+\frac{2HR}{T_d D'^2}\right)\left(e^{-\frac{D'}{2H}t}-1\right) + \frac{R}{T_d D'}t,
\end{equation}
valid $\forall t\in [0,t_n]$. The time instance $t_n$ of frequency nadir is then determined by setting the derivative of \eqref{f(t)} to zero:
\begin{equation}
\label{tn}
    \Delta \dot f(t_n)=0 \longmapsto t_n = \frac{2H}{D'}\ln{\left(\frac{T_d D' \Delta P_L}{2HR}+1\right)}.
\end{equation}
To ensure frequency stability, the frequency nadir has to occur prior to $T_d$, i.e., $t_n\le T_d$, which combined with \eqref{f(t)}-\eqref{tn} yields the expression for frequency nadir $(\Delta f_\mathrm{max}\equiv\Delta f(t_n))$:
\begin{equation}
\label{nadir}
    \Delta f(t_n) = \frac{2HR}{T_d D'^2} \ln{\left(\frac{T_d D' \Delta P_L}{2HR}+1\right)}-\frac{\Delta P_L}{D'}.
\end{equation}
Expressions \eqref{rocof,fss} and \eqref{nadir} indicate that maximum RoCoF and frequency nadir both depend on the aggregate system inertia: former in an inversely proportional fashion and latter through a highly nonlinear function, thus highlighting the potential of synthetic inertia provision for frequency regulation.

\subsection{Approximation of Nadir Constraint}\label{sec:3.3}
After the incorporation of SI provision from WTs under the proposed control framework, the highly nonlinear frequency nadir constraint must be reformulated in order to be included in UC, which is traditionally solved as an MILP. We achieve that by utilizing the simplification from \cite{8667397}, where system damping is initially neglected and subsequently approximated by a linear term. Therefore, the original nadir constraint $\Delta f(t_n)\leq \Delta f_\mathrm{lim}$ can be first transformed into the following form:
\begin{equation}
\label{nadir_no_D}
    HR\ge \frac{\Delta P_L^2T_d}{4\Delta f_\mathrm{lim}}. 
\end{equation}
Since the function $HR=h(D')$ described by \eqref{nadir} is convex and monotonically decreasing, a linear term is included in \eqref{nadir_no_D} to compensate for the contribution of the damping:
\begin{equation}
\label{nadir_c}
    HR\ge \frac{\Delta P_L^2T_d}{4\Delta f_\mathrm{lim}}-\frac{\Delta P_L T_d }{4} D'.
\end{equation}

Note that the feasible region of \eqref{nadir_c} is in general non-convex. Therefore, the problem can not be solved with convex-optimization solvers. The nonlinearity in \eqref{nadir_c} comes from two aspects: the bilinear term $HR$ and the quadratic term $H_{s_j}^2$ in $D'$. To avoid integer variables introduced by binary expansions, a single linearization method is proposed to approximate \eqref{nadir_c} with a set of linear constraints.

\begin{figure}[!t] 
	\centering
	\scalebox{0.5}{\includegraphics[]{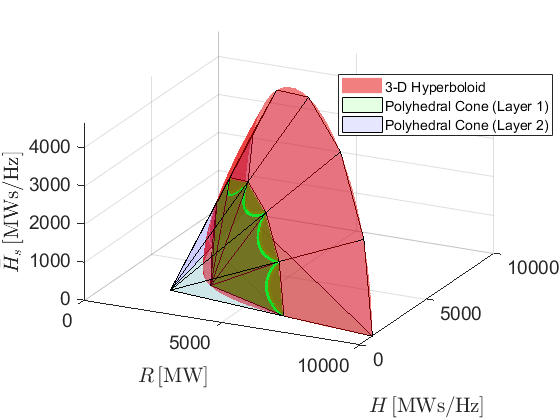}}
	\caption{Sample inner approximation of the frequency nadir constraint using a 3-D hyperboloid with $n=2$ and $m=12$.}
	\label{fig:cone}
\end{figure}

The second-order expression in \eqref{nadir_c} can be rewritten in the standard form of a multi-dimensional hyperboloid using the following linear transformation:
\begin{equation}
\label{LT}
    \begin{bmatrix}H\\R\\ \Tilde{H}_{s}  \end{bmatrix} = \begin{bmatrix}\frac{1}{\sqrt{2}} & \frac{1}{\sqrt{2}} & 0\\ \frac{1}{\sqrt{2}} & -\frac{1}{\sqrt{2}} & 0\\ 0&  0 & I\end{bmatrix} \begin{bmatrix}x_1\\x_2\\ \Tilde{x} \end{bmatrix},
\end{equation}
with $\Tilde{H}_{s} =[H_{s_1},H_{s_2},...,H_{s_{|\mathcal{F}|}}]^\mathsf{T}$, $\Tilde{x} =[x_3,x_4,...,x_{{|\mathcal{F}|}+2}]^\mathsf{T}$, and $I$ being the identity matrix. The substitution of \eqref{LT} into \eqref{nadir_c} results in
\begin{equation}
\label{hyperboloid}
    \sum_{j=1}^{|\mathcal{F}|} \frac{x_{2+j}^2}{\alpha/\beta_j}+\frac{x_2^2}{2\alpha}-\frac{x_1^2}{2\alpha}\le -1,
\end{equation}
and the following coefficients: 
\begin{equation}
    \alpha = \frac{\Delta P_L^2T_d}{4\Delta f_\mathrm{lim}}-\frac{\Delta P_L T_d D}{4},\; \beta_j = \frac{\Delta P_L T_d \gamma_j}{4}.
\end{equation}

Since $x_1\ge 0$, the feasible region of \eqref{hyperboloid} is encompassed by the volume inside the upper sheet of the hyperboloid. Due to its symmetry, the hyperboloid can be inner approximated by $n\cdot m$ hyperplanes where $n,m\in\mathbb{Z}^+$ are the number of parallel layers in the $R-H$ plane and the number of evenly distributed hyperplanes in each layer, respectively. As a result, \eqref{nadir_c} can be approximated by $n\cdot m$ linear constraints of the form:
\begin{equation}
\label{nadir_L}
    a_i H+b_i R+\Tilde{c}_i \Tilde{H}_s+d_i \le 0, \; \forall i\in\mathcal{P},
\end{equation}
with $\mathcal{P}$ being the set of planes (i.e., $|\mathcal P|=n\cdot m$) and $\Tilde{c}_i=[c_{i,1},c_{i,2},...,c_{i,|\mathcal{F}|}]$. Furthermore, it can be proven that the set in \eqref{nadir_c} is convex given $H_{s_j}\ge 0, \forall j \in \mathcal{F}$. Therefore, there exist $\{a_i,b_i,\Tilde c_i,d_i\}$ such that the feasible region of \eqref{nadir_L} is always a subset of that of \eqref{nadir_c}, $\forall i\in \mathcal P$ and $n,m >0$, which guarantees the conservatism of the proposed linear approximation.

An illustrative example of the aforementioned linearization is given in Fig.~\ref{fig:cone}, with a simplified 3-D hyperboloid corresponding to $|\mathcal{F}|=1$ and the original feasible set being replaced with the intersection of $n=2$ polyhedral cones, each described by $m=12$ planes. Moreover, since $\Tilde H_s\ge 0$, only half of the planes are considered in each layer. It should be noted that, in order to achieve better approximation, the density of the layers increases as they approach the vertex of the hyperboloid due to larger curvature. Although any desired level of accuracy could be achieved by increasing the number of planes, the optimal $n$ and $m$ should be chosen to achieve a trade-off between the conservativeness of the model and the associated computational burden.

\subsection{Stochastic Unit Commitment}
\label{sec:3.2}
For the purposes of this paper, a stochastic UC model previously developed in \cite{7833096} is adopted. The renewable energy source uncertainty and generation outages are described by constructing an appropriate scenario tree, whereas the SUC problem minimizes the expected cost over all nodes in the given scenario tree:
\begin{equation}
    \label{eq:SUC}
    \min \sum_{n\in \mathcal{N}} \sum_{g\in \mathcal{G}} \pi (n) C_g(n) 
\end{equation}
A number of constraints are considered such as those of power balance, thermal and storage unit operation as well as system frequency security. \textcolor{black}{Note that for the RoCoF constraint, the maximum instantaneous RoCoF as defined in \eqref{rocof,fss} is considered.} More details regarding the model and SUC formulation can be found in \cite{7833096}. 

\section{Results} \label{sec:4}
Several case studies are conducted based on a GB 2030 power system to demonstrate the effectiveness of the proposed method. The system parameters are set as follows: load demand $P_D \in [20,60] \times10^3\,\mathrm{MW}$, damping $D = 0.5\% P_D / 1\,\mathrm{Hz}$, FR delivery time $T_d = 10\,\mathrm{s}$ and maximum power loss $\Delta P_L = 1800\,\mathrm{MW}$. More details regarding the system description and characteristics of the thermal plants can be found in \cite{8667397}. The frequency limits set by National Grid are: $\Delta f_\mathrm{lim} = 0.8\,\mathrm{Hz}$, $\Delta f_\mathrm{lim}^\mathrm{ss} = 0.5\,\mathrm{Hz}$ and $\Delta \dot f_\mathrm{lim} = 0.5\,\mathrm{Hz/s}$. The annual system operation is simulated under SUC with frequency constraints, which is solved by FICO Xpress through C++ application via BCL \cite{BCL}.

\begin{figure}[!t] 
	\centering
	\scalebox{0.83}{\includegraphics[]{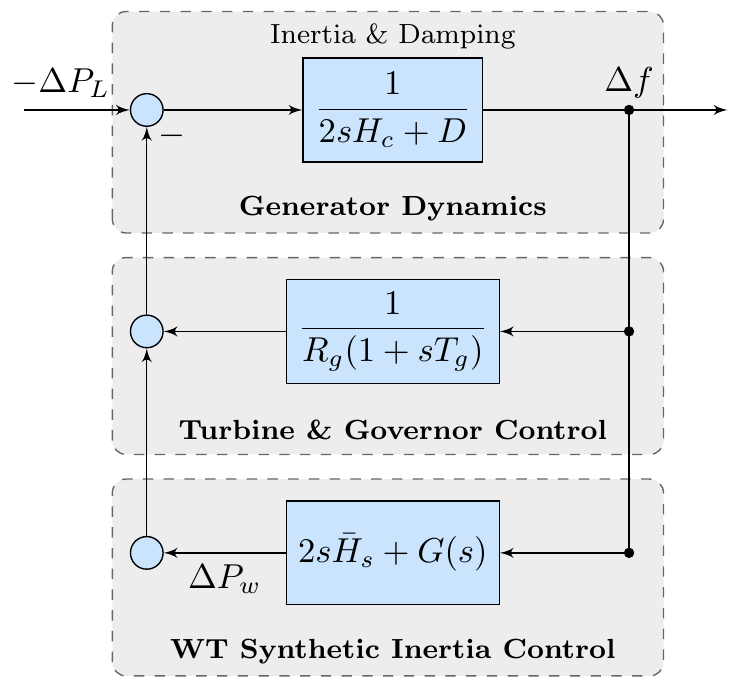}}
	\caption{System frequency dynamics with SI from WTs.}
	\label{fig:sys_dyn}
\end{figure}



\begin{figure}[!b] 
    \centering
    \scalebox{1.1}{\includegraphics[]{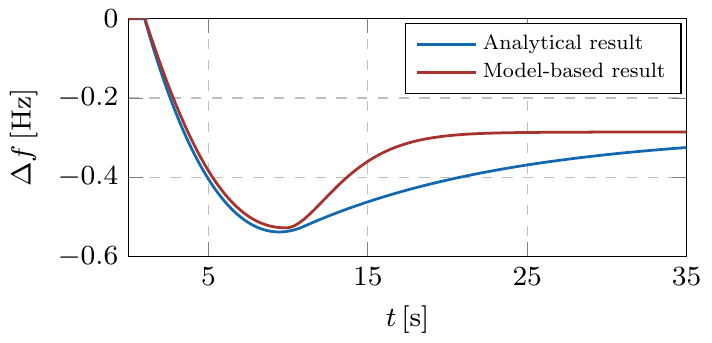}}
    \caption{Frequency comparison of analytical and model-based results.}
    \label{fig:f_compare}
\end{figure}

\begin{figure}[!t]
    \centering
	\scalebox{1.1}{\includegraphics[]{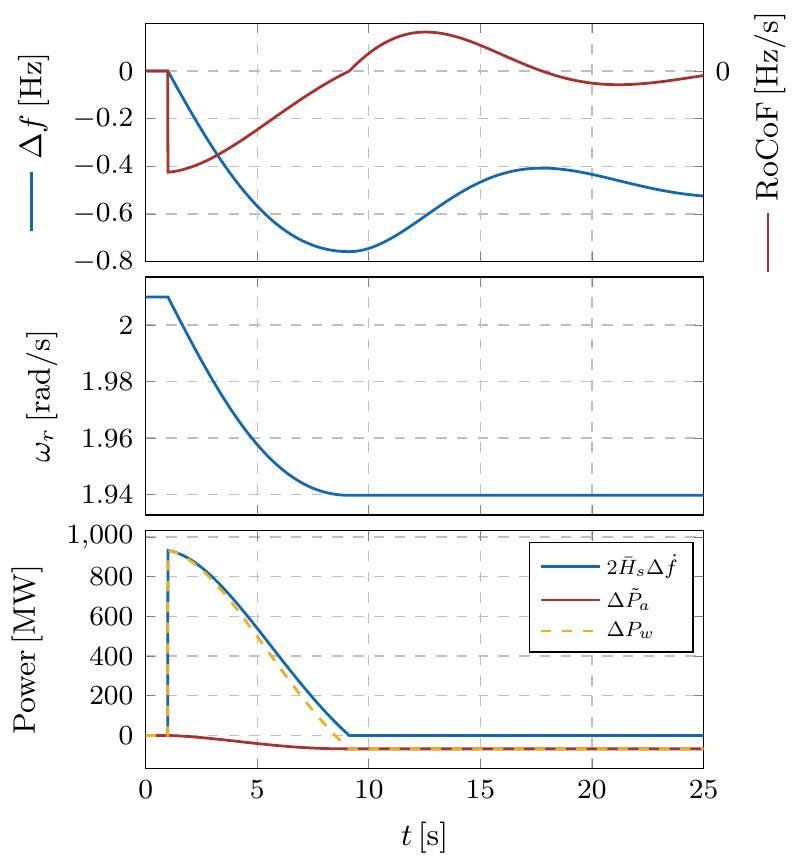}}
	\caption{System response after loss of generation: (i) frequency deviation and RoCoF; (ii) rotor speed; (iii) WT power injection.}
	\label{fig:f&Psi}
\end{figure}

\subsection{Validation of Proposed Frequency \textcolor{black}{Dynamics and} Constraints}\label{sec:4.1}
The accuracy of the derived analytical expressions for the frequency \textcolor{black}{dynamics and} constraints is validated in this section through dynamical simulations using MATLAB/Simulink. The simplified model of system frequency dynamics is illustrated in Fig.~\ref{fig:sys_dyn} where a droop gain $R_g$ and a low-pass filter with time constant $T_g$ represent the dynamics of governor control. The WT SI control is based on the model proposed in Section \ref{sec:2}, with the nonlinear transfer function $G(s) = \Delta \Tilde P_{a}/\Delta f$. 

\textcolor{black}{The modified frequency dynamics described by \eqref{rocof,fss}-\eqref{nadir} are evaluated by comparing the analytical expressions to the model-based results, where the nonlinearities such as synchronous generator dynamics, governor control and wind turbine dynamics are also considered. An accurate approximation of both maximum RoCoF and frequency nadir is observed in Fig.~\ref{fig:f_compare}. Note that the analytical approximation \eqref{f(t)} is valid for the time period until the instance of frequency nadir. Therefore, a noticeable mismatch in the frequency response is present thereafter.}
 
In order to assess the precision of the frequency constraints, \textcolor{black}{a sample solution of the SUC model is provided to the dynamical model resulting in the evolution of CoI frequency and RoCoF depicted in Fig.~\ref{fig:f&Psi}-(i), with the frequency nadir of $0.758\,\mathrm{Hz}$ and maximum RoCoF of $-0.17\,\mathrm{Hz/s}$ being within prescribed limits. A discontinuity is observed at the instance of frequency nadir due to decay of SI provision, resulting in faster frequency convergence and larger RoCoF. The rotor speed illustrated in Fig.~\ref{fig:f&Psi}-(ii) gradually decreases until the frequency reaches its minimum, and subsequently stays constant since the recovery period is not considered.} Moreover, the total WT power injection to the grid ($\Delta P_w$) is presented in Fig.~\ref{fig:f&Psi}-(iii) with its two components - the inertial response and the mechanical power loss - indicated in blue and red respectively. Although the total power injection is less than that of an inertial response alone due to mechanical power loss, the output power reduction is negligible within the first few seconds after the disturbance when SI is most valuable. 

Furthermore, the approximations proposed in Section \ref{sec:2.3} and Section \ref{sec:3.2} for deriving the nadir constraint in \eqref{nadir_L} are assessed and the conservativeness is quantified. This is achieved by obtaining 2952 samples of the linearized nadir expression from the daily SUC dispatch of the GB's 2030 system (corresponding to one month of each season), with system parameters kept within the following ranges: $R\in [1.66,2.71]\,\mathrm{GW}$, $H_s\in [0,3.14]\,\mathrm{GWs/Hz}$ and $H\in [3.71,5.73]\,\mathrm{GWs/Hz}$. For comparison, the actual frequency nadir of each sample is also obtained through dynamic simulations. The results shown in Table \ref{tab:nadir} indicate that the inclusion of frequency constraint \eqref{nadir_L} ensures that the frequency nadir is always maintained below the $\Delta f_\mathrm{lim}$ limit, and the increased number of cones and planes leads to less conservative approximation. However, with the increase in the number of hyperplanes used in the linearization the improvement in the accuracy of the nadir constraint approximation becomes less significant. On the other hand, the computational time grows significantly and might affect the tractability of the optimization problem.  The approximation with $n=4$ cones and $m = 12$ planes has a conservativeness of $0.036\,\mathrm{Hz}$ (i.e., $4.5\,\%$) on average and a reasonable computational time, thus being applied in the remainder of this study.

\begin{table}[!t]
\renewcommand{\arraystretch}{1.2}
\caption{Frequency Nadir Assessment}
\label{tab:nadir}
\noindent
\centering
    \begin{minipage}{\linewidth} 
    \renewcommand\footnoterule{\vspace*{-5pt}} 
    \begin{center}
        \begin{tabular}{ c | c | c | c | c }
            \toprule
             \multirow{2}{5em}{\textbf{Linearization}\\\textbf{Parameters}} &\multicolumn{3}{c|}{\textbf{Nadir} $\boldsymbol{[\mathrm{Hz}]}$}  & \multirow{2}{7.5em}{\textbf{Computational}\\\textbf{Time (Increase)}} \\ 
            \cline{2-4}
            & \textbf{Min}  & \textbf{Mean}  & \textbf{Max}  & \\ 
            \cline{1-5} 
            $n=2, m=6$ & $0.689$ & $0.712$ & $0.742$ & $65.3\,\mathrm{s}\,(0\,\%)$ \\
            \cline{1-5} 
            $n=4, m=12$ & $0.741$ & $0.764$ & $0.791$ & $94.2\,\mathrm{s}\,(44.3\,\%)$  \\
            \cline{1-5} 
            $n=8, m=24$ & $0.741$ & $0.781$ & $0.797$ & $369.3\,\mathrm{s}\,(465.5\,\%)$ \\
           \bottomrule
        \end{tabular}
        \end{center}
    \end{minipage}
\end{table} 

\subsection{Impact of SI Provision on System Operation}\label{sec:4.2}
\begin{figure}[!t] 
	\centering
	\scalebox{1.1}{\includegraphics[]{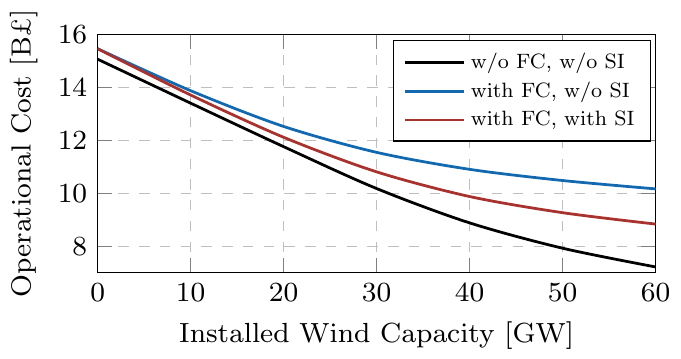}}
	\caption{Impact of frequency constraints and SI provision on the system operation cost.}
	\label{fig:SI_Value}
\end{figure}


\begin{figure}[!t]
  \centering
    \begin{minipage}{0.5\textwidth}
        \centering
        \hspace*{-1.26cm}
        \scalebox{1.1}{\includegraphics[]{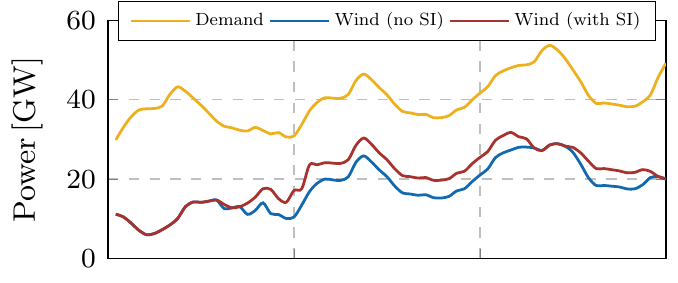}} 
    \end{minipage} 
    \begin{minipage}{0.5\textwidth}
        \centering
        \hspace*{-0.3cm}
        \scalebox{1.1}{\includegraphics[]{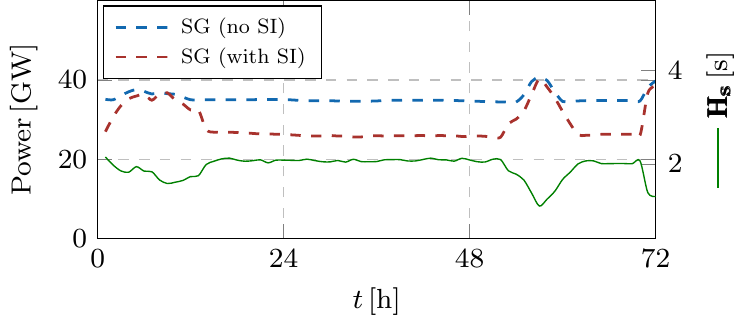}} 
    \end{minipage} 
  \caption{\label{fig:UC_dispatch}Three-day SUC dispatch of (i) demand and wind power production; (ii) synchronous generators and SI provision.}
\end{figure}

This section evaluates the influence of SI provision from WTs on system operation. Three different scenarios are considered: (1) without (w/o) Frequency Constraints (FC) and w/o SI from WTs; (2) with FC and w/o SI; (3) with FC and with SI, with the difference in operational costs summarized in Fig.~\ref{fig:SI_Value} for different levels of installed wind capacity. In the case without frequency constraints (Scenario 1), approximately linear reduction of the operational costs is observed with increase in WT penetration since more energy is supplied by wind generation. With inclusion of FC (Scenario 2) the reduction in operational costs tends to saturate for installed wind capacity above $40\,\mathrm{GW}$ due to an increased number of partially-loaded synchronous generators providing frequency support. Furthermore, the cost of fulfilling frequency requirements - the difference between Scenarios 1 and 2 - increases more than seven times as wind capacity grows from $0$ to $60\,\mathrm{GW}$. By additionally providing SI from WTs (Scenario 3) significant cost savings can be achieved, particularly under high penetration of wind generation, indicated by the cost difference between Scenarios 2 and 3.

The load and wind power profiles as well as aggregated online capacity of Synchronous Generation (SG) for scenarios with and without SI are depicted in Fig.~\ref{fig:UC_dispatch} for the period of three days and installed wind capacity of $60\,\mathrm{GW}$. With SI provision from WTs, the online SG capacity is reduced by $9\,\mathrm{GW}$ in the periods $t\in[14,55]\,\mathrm{h}$ and $t\in [62,70]\,\mathrm{h}$ when there is an abundant wind resource. This is justified by the fact that an increasing amount of SI from WTs reduces the FR requirement from SGs. Similarly, more wind power ($\approx 4\,\mathrm{GW}$), which is curtailed in no SI case to ensure sufficient system inertia, can be utilized under such conditions, as indicated by the difference between the solid red and blue curves. Additionally, the total equivalent SI time constant from WTs \textcolor{black}{as defined in \eqref{H_time_C}},  varying in the range of $[1.1-2.2]\,\mathrm{s}$ is also illustrated in Fig.~\ref{fig:UC_dispatch}, where one notices its inverse relationship with SG capacity. In particular, during the time of low net demand (i.e., $t\in[14,52]\,\mathrm{h}$ and $t\in [64,70]\,\mathrm{h}$) a significant amount of SI is scheduled from WTs and vice versa.

\subsection{Impact of SI Penetration Level and WT Underproduction} \label{sec:4.3}
The influence of SI penetration (i.e., the percentage of WTs with SI capability) and the underproduction of WTs on system operation cost is presented here. As illustrated in Fig.~\ref{fig:Rec}, the cost saving increases considerably for SI penetration of up to $40\,\%$, whereas higher penetration levels do not have a significant contribution to operational costs as total available SI in the system exceeds the necessary capacity. This suggests that it might not be necessary for all WTs in the network to have SI capability. Furthermore, neglecting WT underproduction due to a loss of efficiency would result in an underestimation of system operation costs of up to $170\,\mathrm{M}\text{\pounds}$. With an increasing number of WTs equipped with SI control at high SI penetration levels the SI service provided by an individual WT becomes negligible, thus leading to smaller rotor speed deviation and output shortage. As a result, the impact of overproduction for WT systems with high SI capabilities is significantly reduced. Nevertheless, not taking WT underproduction into consideration during system scheduling process increases the risk of frequency constraint violation and mechanical instability of WTs. This is verified by feeding the results of the SUC model without the WT underproduction into the dynamical model previously described in Section~\ref{sec:4.1}. It is observed that in approximately $70\,\%$ of time instances (2077 out of 2952 hours) a frequency nadir exceeds the limit of $0.8\,\mathrm{Hz}$, $82\,\%$ of which lead to the WT rotor speed falling below the minimum permissible value.

\begin{figure}[!t] 
	\centering
	\scalebox{1.1}{\includegraphics[]{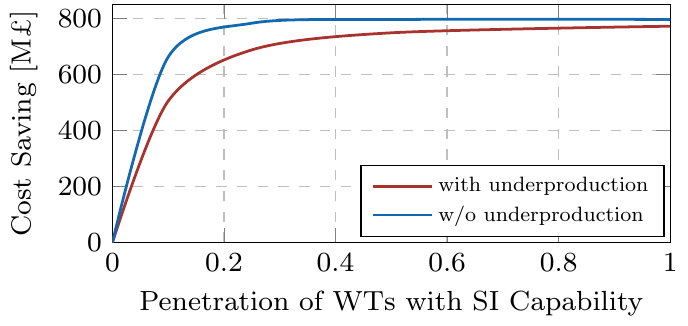}}
	\caption{System Operation Cost Saving under Different SI Penetrations.}
	\label{fig:Rec}
\end{figure}

\subsection{Value of Dynamically Optimizing SI Provision}\label{sec:4.4}

\begin{figure}[!t] 
	\centering
	\scalebox{1.1}{\includegraphics[]{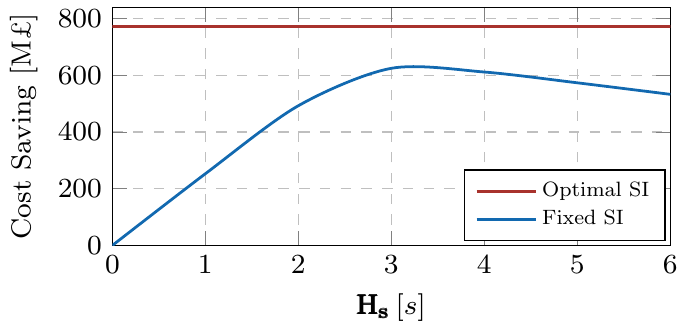}}
	\caption{System Operation Cost Saving with different SI Provision strategies.}
	\label{fig:fix}
\end{figure}

In order to investigate the benefit of optimal SI provision, a fixed value of SI time constant $\pmb{\mathrm{H_s}}$ for all WTs at each hour of the system scheduling horizon is considered as a reference. It can be computed from the total SI $\bar{H}_s$ as:
\begin{equation}
    \label{H_time_C}
    \pmb{\mathrm{H_s}} = \frac{\bar{H}_s f_0}{\sum_{j \in \mathcal{F}}P_{wj}^{c}}.
\end{equation}
It should be noted that for those hours when the maximum system SI is below the fixed value, the maximum SI is applied to avoid infeasibility of the SUC model. In addition, the WT underproduction is also considered for both optimal and fixed SI scenarios in order to demonstrate the cost increase due to excessive SI levels in the system. 

The SI penetration level is set to $50\,\%$ for all scenarios and the impact of the SI provision strategy on operational cost saving is illustrated in Fig.~\ref{fig:fix}. Understandably, the savings pertaining to optimal SI provision are not a function of the fixed SI time constant, as the equivalent WT inertia constant in such case varies from $0$ to $4.16\,\mathrm{s}$ with an average of $2.56\,\mathrm{s}$. 
On the other hand, the system cost saving under fixed SI control increases with higher SI time constants of up to $3.2\,\mathrm{s}$, after which point the benefit of additional SI is diminished by the cost of managing the WT underproduction. 
As a result, the maximum cost saving is roughly $25\,\%$ lower compared to the optimal SI provision case, thus highlighting the value of optimizing SI gains under the proposed control design.

\section{Conclusion and Future Work}    \label{sec:5}
This paper proposes a novel control framework for SI provision from WTs. The mechanical power loss due to the rotor speed deviation is fed back to the electric power reference through MPE such that the secondary frequency dip is eliminated while the effectiveness of frequency regulation is maintained. 
Analytical expressions of relevant system frequency metrics are derived and included as constraints in the SUC model. Moreover,  a set of linear constraints is proposed to approximate the highly nonlinear nadir constraint with a reasonable accuracy and conservativeness. The benefits of optimal SI provision from WTs are clearly identified through several case studies on GB's future power system. The results also suggest that the WT underproduction has a significant impact on frequency security of the system as well as the mechanical stability of the turbine itself, and therefore cannot be neglected during the SI design and system scheduling process.

Future work will extend the proposed model in several directions. In particular, the optimal SI allocation of each WF should be considered, with model predictive control being a potential method for realization of the real-time optimization and implementation. One of the main objectives may involve minimizing the total wind energy loss and the mechanical stress while providing required SI to the system. Furthermore, the CoI-based frequency model should be re-evaluated, i.e., local generator frequencies should be taken into consideration for the purposes of optimal SI control design and system scheduling.

\bibliographystyle{IEEEtran}
\bibliography{bibliography}

\begin{IEEEbiography}[{\includegraphics[trim={.1cm .5cm .45cm .2cm}, width=1in,clip,keepaspectratio]{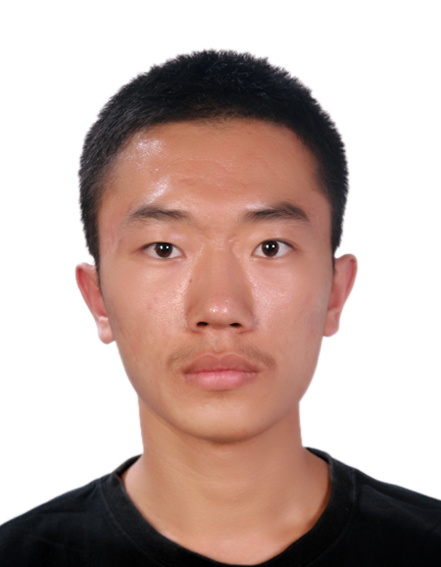}}]{Zhongda Chu}
(S'18) received the M.Sc. degree in electrical engineering and information technology from Swiss Federal Institute of Technology, Zurich, Switzerland, in 2018. He is currently working towards the Ph.D. degree in the Department of Electrical and Electronic Engineering, Imperial College London, London, UK. His research interests include control and optimization of low inertia power systems.
\end{IEEEbiography}
\begin{IEEEbiography}[{\includegraphics[trim={7cm 20cm 20cm 5cm}, width=1in,clip,keepaspectratio]{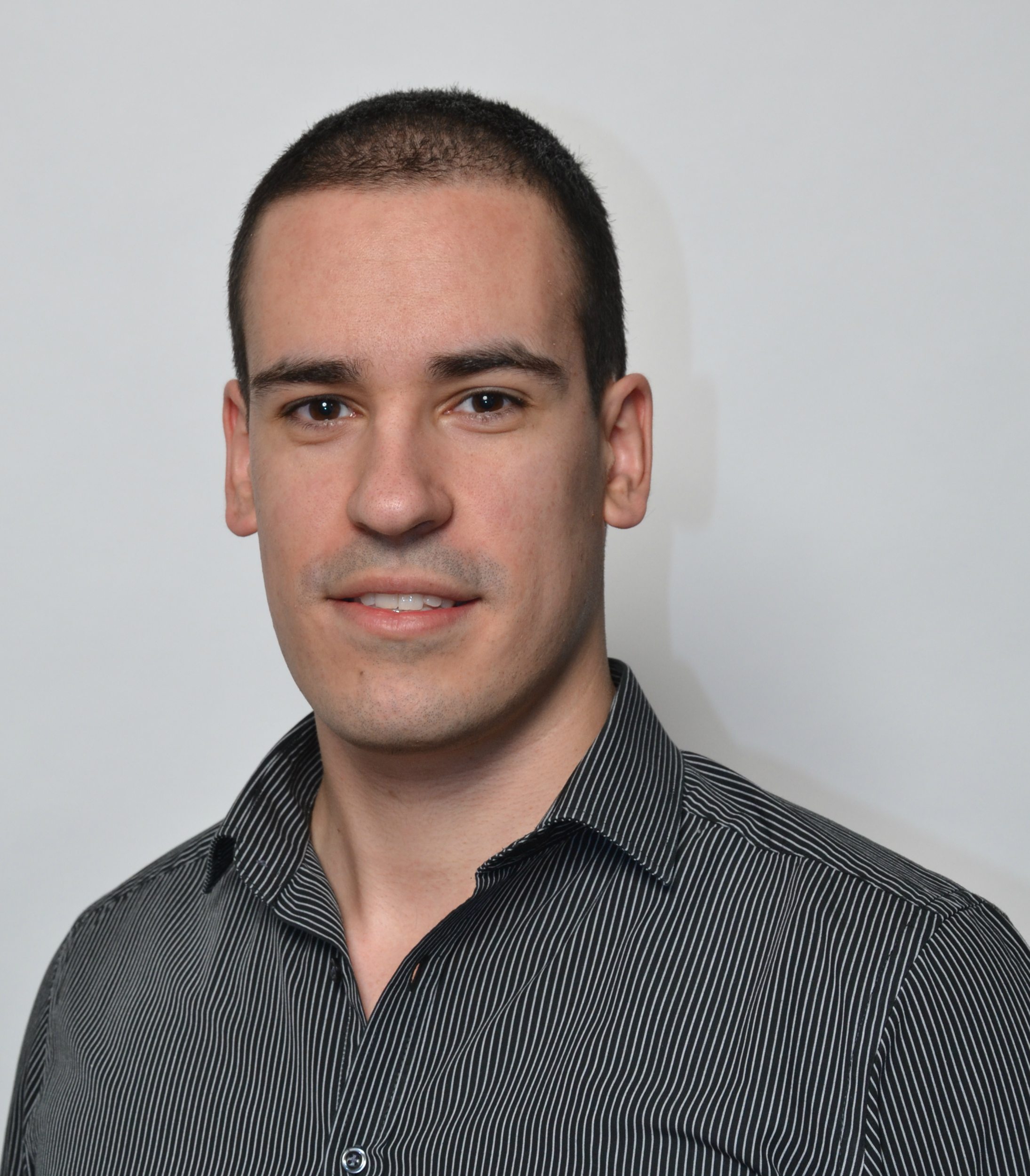}}]{Uros Markovic}
(S'16-M'20) was born in Belgrade, Serbia. He received the Dipl.-Eng. degree in Electrical Engineering from the University of Belgrade, Serbia, in 2013, with a major in power systems. He obtained the M.Sc. and Ph.D. degrees in Electrical Engineering and Information Technology in 2016 and 2020, respectively, both from the Swiss Federal Institute of Technology (ETH), Zurich, Switzerland. He is currently a joint Postdoctoral researcher with the Power Systems Laboratory and the Automatic Control Laboratory of ETH Zurich, Switzerland, and an affiliated researcher at the Grid Integration Group (GIG) of Lawrence Berkeley National Laboratory (LBNL), California, USA. 

His research interests include power system dynamics, control and optimization, with a focus on stability and operation of inverter-dominated power systems with low rotational inertia.
\end{IEEEbiography}
\begin{IEEEbiography}[{\includegraphics[trim={5.5cm 15cm 6cm 4.25cm}, width=1in,clip,keepaspectratio]{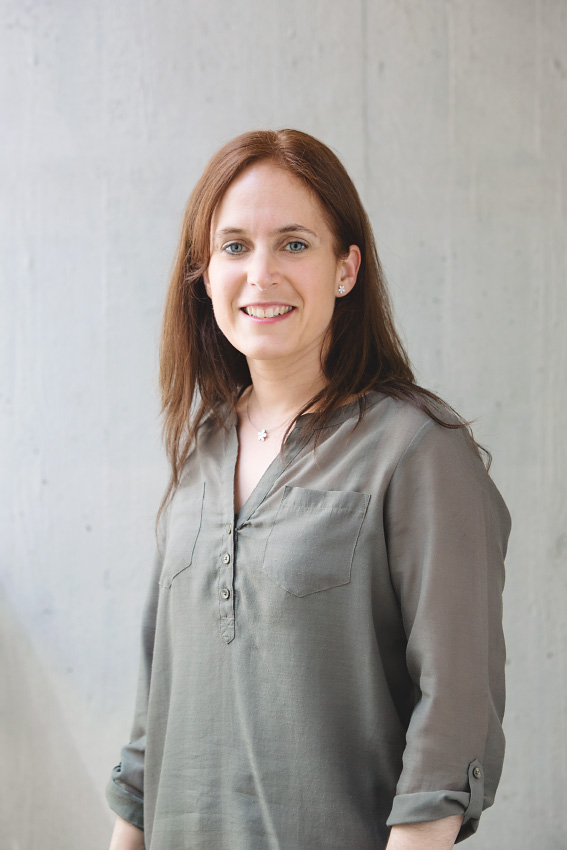}}]{Gabriela Hug}
(S'05-M'08-SM'14) was born in Baden, Switzerland. She received the M.Sc. degree in electrical engineering in 2004 and the Ph.D. degree in 2008, both from the Swiss Federal Institute of Technology (ETH), Zurich, Switzerland. After the Ph.D. degree, she worked in the Special Studies Group of Hydro One, Toronto, ON, Canada, and from 2009 to 2015, she was an Assistant Professor in Carnegie Mellon University, Pittsburgh, PA, USA. She is currently an Associate Professor in the Power Systems Laboratory, ETH Zurich. Her research is dedicated to control and optimization of electric power systems.
\end{IEEEbiography}
\begin{IEEEbiography}[{\includegraphics[trim={.1cm 1.2cm .1cm .1cm}, width=1in,clip,keepaspectratio]{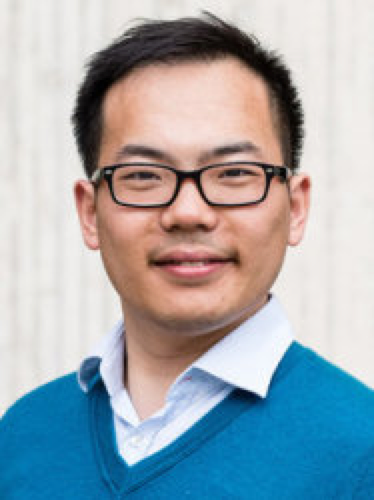}}]{Fei Teng}
(M'15) received the BEng in Electrical Engineering from Beihang University, China, in 2009, and Ph.D. degree in Electrical Engineering from Imperial College London, U.K., in 2015. Currently, he is a Lecturer in the Department of Electrical and Electronic Engineering, Imperial College London, U.K. His research focuses on scheduling and market design for low-inertia power system, cyber-resilient energy system operation and control, and objective-based data analytics for future energy systems.
\end{IEEEbiography}
\vfill


\end{document}


\begin{varwidth}{\linewidth}

\begin{tikzpicture}
\begin{axis}[
    scaled ticks=false,
    tick label style={/pgf/number format/fixed},
    colormap name=viridis,
    width=7.25cm,
    height=4cm,
    xlabel={Installed Wind Capacity  [GW]},
    ylabel={Operational Cost [B\pounds]},
    xmin=0, xmax=60,
    ymin=7, ymax=16,
    xtick={0,10,20,30,40,50,60},
    xmajorgrids=true,
    ymajorgrids=true,
    legend style={at={(axis cs:59.5,15.8)},anchor=north east,nodes={scale=0.75, transform shape}},
    legend cell align={left},
    grid style=dashed,
]
\footnotesize
\addplot[
    smooth,
    thick,
    color=pBlack,
    ]
    table {data/SI_Value/data0.txt};    
    \addlegendentry{\footnotesize w/o FC, w/o SI}       
    
\addplot[
    smooth,
    thick,
    color=pBlue,
    ]
    table {data/SI_Value/data1.txt};        
    \addlegendentry{\footnotesize with FC, w/o SI}  
    
\addplot[
    smooth,
    thick,
    color=pRed,
    ]
    table {data/SI_Value/data2.txt};        
    \addlegendentry{\footnotesize with FC, with SI}  

\end{axis}

\end{tikzpicture} 

\end{varwidth}